\documentclass{article}

\PassOptionsToPackage{numbers,sort&compress}{natbib}
\usepackage[preprint]{neurips_2025}

\usepackage[utf8]{inputenc}
\usepackage[T1]{fontenc}
\usepackage{microtype}
\usepackage{url}
\usepackage{booktabs}
\usepackage{graphicx}
\usepackage{amsfonts}
\usepackage{amsmath}
\usepackage{amssymb}
\usepackage{amsthm}
\usepackage{mathtools}
\usepackage{nicefrac}
\usepackage{xcolor}
\usepackage{xspace}
\usepackage{etoolbox}
\usepackage{algorithm}
\usepackage{algpseudocode}
\usepackage{thm-restate}
\usepackage[most]{tcolorbox}
\usepackage{hyperref}
\usepackage[capitalize,noabbrev]{cleveref}

\newtcolorbox{promptbox}[1][]{
  enhanced,
  colback=black!92,
  colframe=black!92,
  coltext=white,
  boxrule=0.4pt,
  arc=2pt,
  left=8pt,
  right=8pt,
  top=4pt,
  bottom=4pt,
  fontupper=\ttfamily\small,
  fonttitle=\small\ttfamily\bfseries,
  coltitle=white,
  colbacktitle=prompttitle,
  attach boxed title to top left={xshift=6pt, yshift=-2pt},
  boxed title style={colback=prompttitle, frame hidden, sharp corners=south, top=2pt, bottom=2pt, left=4pt, right=4pt},
  #1
}

\definecolor{rolecoluser}{HTML}{61AFEF}
\definecolor{rolecolasst}{HTML}{98C379}
\definecolor{prompttitle}{HTML}{3B5168}
\newcommand{\spkuser}{\textcolor{rolecoluser}{\textbf{User:}}}
\newcommand{\spkasst}{\textcolor{rolecolasst}{\textbf{Assistant:}}}

\newtcolorbox{querybox}[1][]{
  enhanced,
  colback=black!92,
  colframe=black!92,
  coltext=white,
  boxrule=0.4pt,
  arc=2pt,
  left=8pt,
  right=8pt,
  top=4pt,
  bottom=4pt,
  fontupper=\ttfamily\small,
  before upper={\textcolor{green!75!black}{>~}},
  #1
}

\newtcblisting{compactionbox}[1][]{
  enhanced,
  breakable,
  colback=black!92,
  colframe=black!92,
  coltext=white,
  boxrule=0.4pt,
  arc=2pt,
  left=8pt,
  right=8pt,
  top=4pt,
  bottom=4pt,
  fonttitle=\small\ttfamily\bfseries,
  coltitle=white,
  colbacktitle=prompttitle,
  attach boxed title to top left={xshift=6pt, yshift=-2pt},
  boxed title style={colback=prompttitle, frame hidden, sharp corners=south, top=2pt, bottom=2pt, left=4pt, right=4pt},
  listing only,
  listing options={
    basicstyle=\footnotesize\ttfamily\color{white},
    breaklines=true,
    breakindent=0pt,
    columns=fullflexible,
    keepspaces=true,
    upquote=true,
  },
  #1
}

\theoremstyle{plain}
\newtheorem{theorem}{Theorem}

\newtheorem{corollary}[theorem]{Corollary}

\newtheoremstyle{defnstyle}%
  {\topsep}{\topsep}%
  {}{}%
  {\bfseries}{.}%
  { }%
  {\thmname{#1}\thmnumber{ #2}\thmnote{ (#3)}}
\theoremstyle{defnstyle}
\newtheorem{definition}{Definition}

\AtBeginEnvironment{definition}{\pushQED{\qed}}
\AtEndEnvironment{definition}{\popQED}

\newcommand{\calX}{\mathcal{X}}

\newcommand{\calQ}{\mathcal{Q}}

\newcommand{\calI}{\mathcal{I}}
\newcommand{\calT}{\mathcal{T}}

\newcommand{\err}{\mathrm{err}}
\newcommand{\val}{\mathrm{val}}

\title{Context Compaction Theory}

\author{%
  Hayder Tirmazi \\
  AllSpice Inc. and Boston University \\
  \texttt{hayder@bu.edu} \\
  \And
  Sam Markelon \\
  Proof Trading \\
  \texttt{sam@prooftrading.com} \\
  \AND
  Allison Bishop \\
  Proof Trading and City College of New York \\
  \texttt{abishop@ccny.cuny.edu} \\
  \And
  Michael Mitzenmacher \\
  Harvard University \\
  \texttt{michaelm@eecs.harvard.edu} \\
}

\begin{document}

\raggedbottom

\maketitle

\begin{abstract}
Large Language Models (LLMs) have a bounded context window. The context window is the maximum input size an LLM can consume for a single inference. AI agents rely on a process called context compaction to fit their state within the context window when calling an LLM. Despite its ubiquity, context compaction has received essentially
no formal analysis. In this paper, we initiate a formal study of context compaction. We first introduce a framework consisting of two games that capture the two algorithmic strategies for context compaction used by contemporary AI agents in practice. The Context Selection Game models context compaction algorithms that select a subset of an agent's accumulated state to retain. The Context Generation
Game models context compaction algorithms that summarize an agent's state by an arbitrary message of bounded
length. We then prove an equivalence between the Context
Generation Game and one-way communication complexity. The minimum context compaction budget for answering a set of queries within a target error is equal to the one-way communication complexity of the induced communication problem at the same error. Known bounds from communication complexity therefore transfer directly to context compaction. We also show that the Context Selection Game corresponds to a restricted class of one-way communication protocols. Any gap between selection and generation is therefore a gap between two classes of communication protocols. We prove that there exists a set of queries for which generation needs strictly less budget than selection. The equivalence between the Context Generation Game and one-way communication also lets us measure how well a deployed context compaction algorithm performs on a query relative to the optimal strategy. For example, we
develop this comparison on set membership queries, where a Bloom filter
is a near-optimal context compaction algorithm.
We present a case study that evaluates Anthropic's context compaction endpoint on set membership queries. In our case study, the endpoint answers membership queries with a substantially higher error rate than a Bloom filter of the same size.
\end{abstract}


\section{Introduction}

AI agents are increasingly used to assist with complex tasks such as
writing code~\citep{anthropic2025claudecode, google2025geminicli, cursor2024cursor, stripe2026minions},
debugging systems~\citep{stripe2026minions}, and
conducting research~\citep{schmidgall2025agent, zhu2026paperbanana, anthropic2025multiagent}.
An agent is an application that accomplishes a task by repeatedly calling a Large Language Model (LLM) and executing the actions that the LLM requests. Over the course of a task, the agent accumulates a state. The state includes the messages the user has sent, the responses of the LLM, and the results of actions such as reading a file.

An LLM accepts an input of bounded size, and the agent must construct this input on every call. The bound on the input size is called the LLM's context window. When the agent's state is small, the agent sends its entire state to the LLM. When the agent's state grows beyond the context window, the agent must construct a smaller input that preserves the information the LLM will need. This process is called context compaction~\citep{anthropic2025compactioncookbook, claude2026compaction}. Figure~\ref{fig:agent-llm} shows the relationship between the agent, its internal state, and the LLM. The constraint that drives context compaction is the bounded input of the LLM. The size of the agent's memory is not the constraint. Agents also face a second, softer pressure. The reasoning quality of an LLM degrades as its input grows. This phenomenon is called context rot~\citep{zhang2025recursive}. Context rot pushes agents to do context compaction even before their state reaches the context window. This smaller usable limit is the effective context window.

\begin{figure}[t]
\centering
\includegraphics[width=0.85\linewidth]{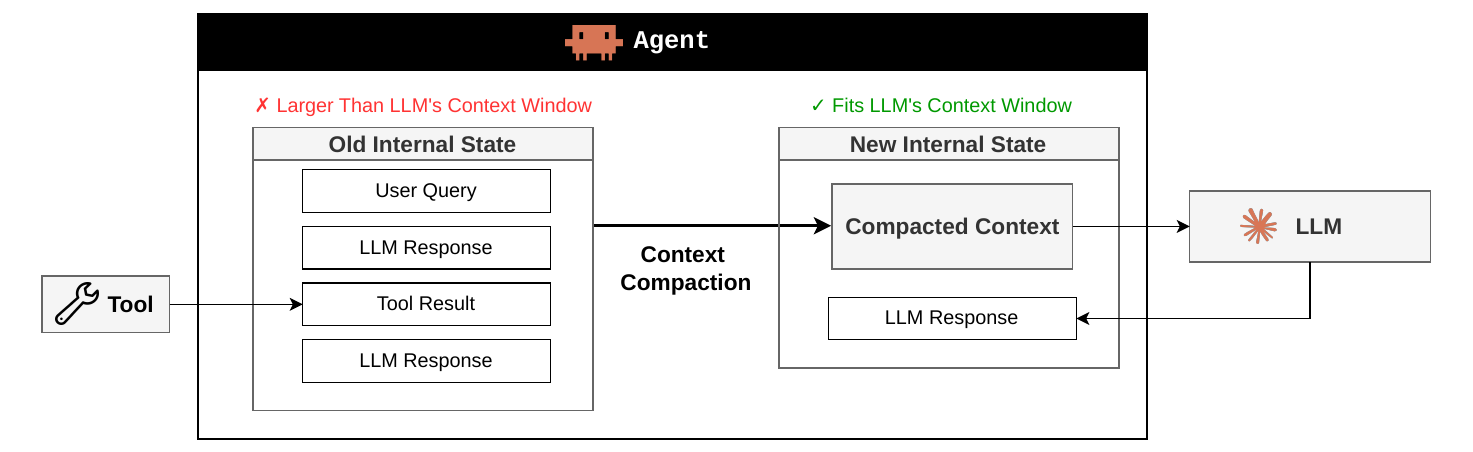}
\caption{Context compaction in an agent. The agent stores an internal state that grows as it works. The internal state holds the user's queries, the LLM's responses, and the results of tool calls. The LLM has a bounded input size called the context window. When the internal state exceeds the context window, the agent runs a context compaction algorithm that produces a compacted context. This compacted context, together with the LLM's next response, becomes the agent's new internal state, which fits within the context window. The agent carries this new internal state forward.}
\label{fig:agent-llm}
\end{figure}

During context compaction, an agent summarizes its state or selects a subset of it.
Systems such as Claude Code~\citep{anthropic2025claudecode}, Cursor~\citep{cursor2024cursor}, and Codex~\citep{openai2025codex} all implement variants of this strategy,
typically triggered when context usage exceeds a threshold, e.g., 95\% of the
allowed context window. Context compaction itself introduces a new failure
mode. After a context compaction, the agent replaces its internal state 
with the compacted context and only sends the compacted context 
to the LLM in later steps of the session. Therefore, information a context compaction discards no longer reaches the LLM in future steps of the session. This can cause the agent to make decisions
inconsistent with the original conversation. Note that the agent can and
does maintain a record of the full history in storage. However, the agent does not act on the full
history going forward, because working from the full history
instead of replacing it with the compacted context would require
re-running a context compaction algorithm over the full history on every
query. This alternative is slow and costly, as we show in
Section~\ref{subsec:agents}. Despite its ubiquity, context compaction has received essentially
no formal analysis. Existing work on context compaction mostly focuses on
empirical evaluations~\citep{jiang2023llmlingua, kang2025acon, li2026latent} and
does not provide formal guarantees on what information is preserved or lost.

In this paper, we initiate a formal study of context compaction with the following contributions.

\paragraph{A theoretical framework for context compaction.}
In Section~\ref{sec:model}, we introduce two games for analyzing context
compaction. The Context Selection Game formalizes context compaction by
selection, in which a selector chooses a subset of items to retain subject
to a budget constraint, and an adversary then issues a query that must be
answered using only the selected subset. The Context Generation Game extends the framework to
context compaction by generation, in which the context compaction algorithm emits an
arbitrary message of bounded length, and the query must be answered from
that message. Both games model a single invocation of context compaction. An agent session invokes context compaction many times. We relate the single-invocation model to the full agent loop in Section~\ref{subsec:agents}. We classify production context compaction algorithms used by Codex,
Claude Code, Gemini CLI, and OpenCode as instances of these two games.

\paragraph{Generation as one-way communication.}
In Section~\ref{sec:communication}, we prove an equivalence between the Context Generation
Game and one-way communication (Theorem~\ref{thm:gen-equiv}).
More precisely, we show that for any set of queries, the minimum budget at which a context generation algorithm
achieves error $\leq \epsilon$ equals the one-way randomized
communication complexity of the induced communication problem at error
$\epsilon$. The equivalence lets us import communication-complexity results into
the context compaction setting. We also show
that the Context Selection Game (Corollary~\ref{cor:select-restricted}) corresponds
to a restricted protocol class within one-way communication. We exhibit a family of queries on which selection needs a factor of $\Theta(\log n)$ more budget than generation to reach the same error (Theorem~\ref{prop:select-gen-gap}).

\paragraph{A case study of deployed context compaction.}
In Appendix~\ref{app:empirical}, we present a case study that uses our equivalence to measure a context compaction endpoint deployed in production. We record a set of $15{,}000$ items in the context of Anthropic's context compaction endpoint. We tell the endpoint that its result will only be used for set membership queries. We then ask membership queries and compare the error rate to a Bloom filter of the same compacted size. Note that a Bloom filter is an optimal context compaction algorithm for set membership up to a constant factor. We find that the endpoint answers membership queries with an error rate close to a random guess. We show using a control experiment that information lost during context compaction causes this error.

\subsection{Related Work}
\label{subsec:related-work}

Context compaction relates to several lines of work on managing what an LLM reads and remembers, most of them empirical or systems-oriented. We review the most relevant below and, in each case, note how it differs from our formal study of what context compaction must retain.

\paragraph{Context Compaction}
Most prior work on context compaction is empirical. Prior work evaluates prompt compression
methods~\citep{jiang2023llmlingua, jiang2024llmlingua2} and context compaction methods in agents~\citep{kang2025acon, li2026latent} by downstream task accuracy. These methods do not provide formal guarantees on what
information remains in an agent's compacted context. Context compaction methods in production agents such as Claude
Code~\citep{anthropic2025claudecode}, Codex~\citep{openai2025codex}, Gemini
CLI~\citep{google2025geminicli}, and OpenCode~\citep{opencode2026} 
are typically based on LLM summarization and are also largely guided by empirical evaluation.

\paragraph{Inference-level Context Management}
A separate line of work acts inside the LLM's inference engine. This includes KV-cache eviction methods such as H2O~\citep{zhang2023h2o} and StreamingLLM~\citep{xiao2024streamingllm} that drop low-value tokens from the attention cache so that a long context fits in fixed memory. \citet{liu2023lost} show that LLMs use long context unevenly and accuracy drops when the relevant information sits in the middle of a long context. These inference-level methods are not available to most agent developers for three reasons. Firstly, agents often access an underlying LLM only through externally hosted inference providers such as AWS Bedrock~\citep{awsbedrock}. Secondly, the LLM an agent uses may be proprietary, making it difficult to modify its inference implementation. Finally, modern agents often switch among several LLMs, using them as interchangeable backends independent of the agent's design~\citep{shahout2026orlalibraryservingllmbased}. The primary lever an agent developer controls is the context it constructs and sends to the LLM. Agents therefore rely on context compaction. Note that LLM inference-level context management methods are complementary to context compaction methods, and an agent takes advantage of both.

\paragraph{LLMs and Set Membership}
Recent work studies how reliably LLMs answer set membership questions~\citep{hergert2025brittleness, guo2026hallucination}.
\citet{hergert2025brittleness} find that instruction-tuned LLMs are brittle even
on explicit membership queries. Closest to our setting, \citet{guo2026hallucination} show that an LLM storing
items in bounded memory trades off forgetting against hallucinating similar to
a Bloom filter trading space against false positives. \citet{guo2026hallucination} study the memory an LLM
acquires during training. Our work, on the other hand, studies context
compaction at inference time. This is relevant in the agentic setting because an
agent's task usually depends on information the LLM never saw in training,
such as the current codebase, the files it has edited, and its tool outputs. This
information exists only in the context that the agent sends to the LLM.
Context compaction shrinks that context, and therefore determines how much of this
information the LLM can use.
Approximate membership is also a popular benchmark for learned and neural
data structures~\citep{kraska2018case, mitzenmacher2018model, rae2019metalearning}.

\section{Preliminaries}
\label{sec:prelim}

In this section, we provide an overview of how production agents work
and why context compaction is necessary for long-running agentic
sessions. We use OpenAI's Codex~\citep{openai2025codex} v0.125.0
as our primary reference agent for this discussion. Codex is one of the most popular agents that is also open-source, which allows us to analyze its implementation in detail. However, Codex is not unique in its design. The vast majority of popular agents today including Claude Code~\citep{anthropic2025claudecode}, OpenCode~\citep{opencode2026}, Block's Goose~\citep{block2025goose}, and Gemini CLI~\citep{google2025geminicli} implement the same core agent loop design as Codex. Therefore, our analysis of Codex's design and context compaction mechanism is representative of a broad class of popular agents. We mention other agents including Claude Code, OpenCode, and Gemini CLI only when their designs differ in meaningful and informative ways.

The remainder of this section is organized as follows. Section~\ref{subsec:llms-tools} introduces LLMs and the
tool-calling interface used by modern agents.
Section~\ref{subsec:agents} describes a canonical agent loop and the role of
context compaction within it.
Section~\ref{subsec:example} works through a concrete example agent session
using Codex.
Section~\ref{subsec:existing-compaction} surveys the context compaction
algorithms used by Codex, Claude Code, Gemini CLI, OpenCode, and others.

\subsection{LLMs and Tools}
\label{subsec:llms-tools}

Tokens are the basic units of text that LLMs process. They are usually words or subword units within a fixed vocabulary. An LLM is a function that takes a sequence of tokens as input and produces another sequence of tokens as output. This process is called inference. Agents build on LLM inference to perform specific tasks. For a given LLM inference, the input sequence of tokens is called the context of the LLM for that inference. Every LLM has a hard architectural limit on the number of input tokens it can process during a single inference. This hard limit on the input token count is called the LLM's context window. We denote this limit by $W$. Inputs longer than $W$ tokens cannot be processed by the LLM during a single inference. In the rest of this paper, when we refer to context, we mean specifically the input token sequence sent to the LLM during inference. This is the standard way in which practitioners use the term context in the agentic setting.

In addition to the hard context window limit $W$, recent work has documented a phenomenon known as context rot~\citep{zhang2025recursive} which degrades the reasoning quality of LLMs as the context length grows. Due to context rot, there is often a cutoff point even smaller than the hard limit $W$ where the LLM's reasoning quality degrades in a way that is no longer desirable. This cutoff point is the effective context window, often much smaller than $W$. This means context window size is often a concern for agents even for sessions that are short enough to fit within the hard limit $W$.

Modern LLMs also allow agents to specify custom tools. The agent provides these tools. Its developer, its user, or a third-party developer can write them. Agents typically use tools to perform tasks outside the LLM's capabilities, such as file reads, shell commands, and web search. In some systems, the tools might even be other LLMs~\citep{zhang2025recursive} or other agents~\citep{langgraph}. The agent registers a set of tool schemas which are typically structured JSON definitions describing each tool's name, parameters, and behavior. The agent sends the tool schemas to the LLM as context during inference. When the LLM wants to use a tool, instead of producing plain text it produces a structured tool call request containing the tool name and its arguments. The agent receives these tool call requests. It executes the tool and includes the result in its next input to the LLM.

An agent can also register a system prompt. The system prompt is a static block of instructions describing the LLM's role, operating rules, and any task-specific guidance. The agent prepends the system prompt to the context of every LLM inference in a session. As a concrete example, the first few lines of Codex's system prompt for GPT-5 read:
\begin{promptbox}[title=An excerpt from Codex's system prompt for GPT-5]
\small\ttfamily
You are Codex, based on GPT-5. You are running as a coding agent in the Codex CLI on a user's computer.

\#\# General

- When searching for text or files, prefer using `rg` or `rg --files` respectively because `rg` is much faster than alternatives like `grep`. (If the `rg` command is not found, then use alternatives.)
\end{promptbox}
The full Codex system prompt continues for several hundred lines, covering editing constraints, planning, sandboxing, and other operating rules.

\subsection{Agents}
\label{subsec:agents}

As mentioned above, an agent is an application that relies on LLM inference. An agent coordinates three core operations. The first operation is invoking an LLM. The second operation is executing one or more tools. The third operation is updating its own internal state. The internal state of an agent at iteration $t$, denoted $\calI_t$, is the working state the agent carries about its task and the world. It includes 1) messages sent to the agent, 2) the LLM's responses, 3) tool calls and their results, and 4) anything programmatically tracked, computed, or retrieved by the agent's own code that does not fall into the other categories. To call the LLM, the agent must present its internal state as an input of at most $W$ tokens, where $W$ is the size of the LLM's context window. When $\calI_t$ fits within the context window, the agent sends it to the LLM directly. When $\calI_t$ exceeds the context window, the agent invokes a context compaction algorithm that replaces $\calI_t$ with a bounded internal state that encodes the relevant information and satisfies $|\calI_t| \leq W$. Note that the agent's own memory is not the constraint here. The agent could store everything it has gathered. In fact, agents often store the full transcript of their session. The constraint is that the LLM that the agent uses accepts at most $W$ tokens of input. An alternative would be for the agent to recompute a summary from the stored session history on every call. However, this design balloons the agent's cost per user query in terms of LLM usage as well as the time taken to serve each user query. Agents commonly do context compaction by asking the LLM to summarize the history. This method is called LLM-based summarization. Consider a coding session with $800{,}000$ tokens of history. Claude Fable~5 has a context window of $1{,}000{,}000$ tokens~\citep{anthropic2026pricing}. This history fills most of Claude Fable~5's context window. Suppose the agent re-derives a $100{,}000$-token summary from the history with LLM-based summarization. The summarization call reads $800{,}000$ input tokens and generates $100{,}000$ output tokens. Fable~5 charges \$10 per million input tokens and \$50 per million output tokens~\citep{anthropic2026pricing}. The call therefore costs \$8 in input and \$5 in output. That adds up to \$13 for a single query. Fable~5 generates about $65$ output tokens per second~\citep{artificialanalysis2026fable}. That is about $15$ milliseconds per output token. Generating the $100{,}000$-token summary alone takes about $26$ minutes. Reading the $800{,}000$-token input adds further delay. Note that prompt caching~\citep{anthropic2026promptcaching} can potentially reduce the input cost in our calculation, but even then the output cost is unaffected. As the history grows over the session, this cost and delay increase monotonically. Most deployed agents therefore feed only the compacted state to the LLM in later steps of the agent's session, even though they keep the raw history in storage. This is the context compaction problem we study in this paper.

An agent operates in two nested loops. The outer loop processes user
messages over the course of a session. The inner loop processes a
single user message via repeated LLM calls and tool executions, until
the LLM produces a response with no further tool calls. The initial
internal state is $\calI = \{P\} \cup \calT$, where $P$ is the system
prompt and $\calT$ is the set of tool schemas. The full loop is given
in Algorithm~\ref{alg:agent-loop}.

\begin{figure}[t]
\begin{minipage}[t]{0.56\linewidth}
\begin{algorithm}[H]
\caption{An agent's main loop}
\label{alg:agent-loop}
\begin{algorithmic}[1]
\State $\calI \gets \{P\} \cup \calT$
\Loop
    \State $u \gets \textsc{ReceiveUserMessage}()$
    \State $\calI \gets \calI \cup \{u\}$
    \Repeat
        \State $\calI \gets \textsc{Compact}(\calI)$
        \State $(R, T) \gets \text{LLM}(\calI)$
        \State $\calI \gets \calI \cup \{R\}$
        \If{$T \neq \emptyset$}
            \State $E \gets \{\textsc{Execute}(\tau) : \tau \in T\}$
            \State $\calI \gets \calI \cup E$
        \EndIf
    \Until{$T = \emptyset$}
    \State Show $R$ to the user
\EndLoop
\end{algorithmic}
\end{algorithm}
\end{minipage}\hfill
\begin{minipage}[t]{0.40\linewidth}
\begin{algorithm}[H]
\caption{Single invocation setting}
\label{alg:single-compaction}
\begin{algorithmic}[1]
\State $\calI \gets$ the agent's internal state
\State $\calI \gets \textsc{Compact}(\calI)$
\State $q \gets$ a future user query
\State $R \gets \text{LLM}(\calI \cup \{q\})$
\end{algorithmic}
\end{algorithm}
\end{minipage}
\end{figure}

Line 1 initializes the internal state $\calI$ with the system prompt
$P$ and the tool schemas $\calT$. The outer loop on lines 2--15
iterates over user messages received during the session, and the inner
loop on lines 5--13 processes a single user message via repeated LLM
inferences and tool executions. On each iteration of the outer loop,
line 3 waits for a new user message $u$ and line 4 adds the message to the
agent's internal state. On each iteration of the inner loop, line 6 invokes
$\textsc{Compact}$ to fit the internal state within the context window, so that
$|\calI| \leq W$.
Then line 7 conducts LLM inference on the internal state to obtain a response $R$ and a possibly
empty set $T = \{\tau_1, \ldots, \tau_n\}$ of tool call requests. Line 8 adds the response $R$ to the internal state. If $T$ is non-empty,
lines 10 and 11 execute the tool calls requested by the LLM and add their results $E$ to
the agent's internal state. The results of the tool calls are then communicated to the LLM during the next iteration of the inner loop. The inner loop terminates when the LLM produces a
response with no tool calls. Line 14 then shows that response to the
user, and the outer loop waits for the next user message. Note that the inner loop is often referred to as a reasoning and acting (ReAct) loop~\citep{yao2022react,ibmreact}.

The $\textsc{Compact}$ routine on line 6 is the focus of this paper.
The $\textsc{Compact}$ routine reduces the internal state to at most
$W$ tokens. When the agent's internal state already fits within the context
window, the agent can pass it through to the LLM unchanged. When the
agent's internal state exceeds $W$ tokens, the agent must perform context
compaction. The tool results $E$ added to the internal state on line
11 may include file contents, command output, error messages, and
other tool-specific outputs. Between context compactions, the internal state grows as
the agent appends responses and tool results. Each invocation of the context compaction algorithm 
resets the internal state to a context of at most $W$ tokens. A single user message can trigger
dozens of inner iterations, and a session can contain many user
messages.

In Algorithm~\ref{alg:agent-loop}, the agent carries the compacted context
forward. After a context compaction, the agent uses the compacted context, the new
response, and any new tool results as its new internal state.
When the agent does context compaction again, it operates on a state that already contains the output
of earlier context compactions. This matches deployed agents including Codex~\citep{openai2025codex}, Claude
Code~\citep{anthropic2025claudecode}, Gemini CLI~\citep{google2025geminicli},
OpenCode~\citep{opencode2026}, Goose~\citep{block2025goose}, and
Cursor~\citep{cursor2024cursor}. All of these agents feed the compacted context to the LLM in
place of the full history, and each later context compaction builds on the previous
summary. We discuss more details on context compaction mechanisms in existing agents in
Section~\ref{subsec:existing-compaction}. As we explained above, recomputing a fresh context from the full history on every turn would be expensive. This cost is why deployed agents carry the compacted context forward.

Our model captures a single invocation of context compaction, shown as
Algorithm~\ref{alg:single-compaction} beside the full loop. The agent does context compaction on
its internal state, producing a bounded context. A query then arrives and must be
answered from that context alone. Since the context compaction is committed before the
query arrives, the context compaction algorithm cannot depend on the query. We study a
single invocation because it is the necessary first step toward understanding a
full agent session. An agent does context compaction many times. Each context compaction runs on the output of an earlier
one and can only discard more information. A single invocation is therefore the
base case. It is also the most favorable case for the agent. If a single
context compaction cannot preserve an answer, neither can a session of many context compactions.
Repeated context compaction is a natural extension of this setting. We leave its analysis
as an open problem in Section~\ref{sec:discussion}.

In the case of OpenAI's Codex specifically, the agent loop is organized into ``turns''. The vast majority of currently popular agents including Claude Code~\citep{anthropic2025claudecode}, OpenCode~\citep{opencode2026}, Block's Goose~\citep{block2025goose}, and Gemini CLI~\citep{google2025geminicli} implement the agent loop in a similar manner, though with different data structures and implementation details.

\subsection{Example}
\label{subsec:example}

In this section, we provide a concrete example of an agent, in our case OpenAI's Codex, operating in its main loop. Consider a session where Codex is given the following user message $u_0$.

\begin{querybox}
Can you fix the TypeError in main.py?
\end{querybox}

where \texttt{main.py} is a Python file that contains a function that raises a \texttt{TypeError}. We assume the system prompt $P$ is the default system prompt for Codex's GPT-5 model and $\calT$ is the default set of tool schemas. We use two of Codex's built-in tools. The first tool is \texttt{shell} which executes a shell command and returns its standard output. The second tool is \texttt{apply\_patch} which edits one or more files by applying a patch.

\medskip\noindent The agent loop proceeds as follows.

\begin{description}
\setlength\itemsep{0.4em}
\item[Iteration 0.] The agent sends its internal state $\calI_0 = \{u_0, P\}\, \cup \,\calT$ to the LLM. The LLM responds with $R_0 = $ ``Let me read the file'' and a single tool call $T_0 = \{\texttt{shell(``cat main.py'')}\}$. The agent executes the tool and obtains $E_0$, the contents of \texttt{main.py}. The agent's internal state becomes $\calI_1 = \calI_0 \cup \{R_0\} \cup E_0$.
\item[Iteration 1.] The agent sends $\calI_1$. The LLM responds with $R_1 = $ ``I see the issue on line 42'' and a tool call $T_1 = \{\texttt{apply\_patch(\ldots)}\}$ patching \texttt{main.py}. After execution, $\calI_2 = \calI_1 \cup \{R_1\} \cup E_1$ where $E_1$ is the patch confirmation.
\item[Iteration 2.] The agent sends $\calI_2$. The LLM responds with $R_2 = $ ``Let me verify by running the tests'' and a tool call $T_2 = \{\texttt{shell(``pytest'')}\}$. The pytest output is large, so $E_2$ is large. After the update, $|\calI_3|$ exceeds the effective context window, the threshold at which the agent triggers context compaction.
\item[Iteration 3.] The agent therefore does context compaction on $\calI_3$ so that it fits within the context window. The agent's choice of what to keep determines whether it retains the bug fix decision $R_1$, the test failure details $E_2$, or the original user goal $u_0$.
\end{description}

In this example, Iteration 3 is the first iteration at which the agent has to perform context compaction. The agent must choose what information to retain and what information to discard in order to fit the context window constraint. If the agent retains $R_1$ but discards $E_2$, then the agent may forget the test failure details and fail to fix the bug. If the agent retains $E_2$ but discards $R_1$, then the agent may forget that it already decided on a fix and fail to apply it. If the agent retains $u_0$ but discards both $R_1$ and $E_2$, then the agent may forget both the bug fix decision $R_1$ and the test failure details $E_2$, and fail to make any progress on fixing the bug. This example illustrates how context compaction can lead to loss of information that is relevant for future reasoning steps, which motivates our formal analysis of context compaction in this paper.

\subsection{Existing Context Compaction Mechanisms}
\label{subsec:existing-compaction}

OpenAI's Codex maintains a per-model token threshold that sits below the context window. When token usage crosses this threshold, an automatic context compaction task runs while the conversation still fits within the context window. The context compaction task assembles the entire conversation history together with a fixed context compaction prompt and sends it to the LLM. The prompt instructs the LLM to produce a summary of the existing conversation history. The LLM's summary then replaces the conversation history for future iterations of the agent loop. This is LLM-based summarization. After context compaction, Codex warns the user that ``long threads and multiple compactions can cause the model to be less accurate,''~\citep{openai2025codex} explicitly acknowledging that repeated context compaction degrades quality. Note that Codex's context compaction mechanism is independent of the underlying LLM. It is identical whether the LLM is a frontier LLM accessed via an API or a small open-weight LLM running locally. Therefore, the context compaction problem is a property of the agent design rather than of any particular LLM provider.

The vast majority of popular agents today make use of LLM-based summarization in the same way as Codex. Some agents use it as their only strategy, similarly to Codex. For example, Gemini CLI~\citep{google2025geminicli} triggers context compaction at a configurable percentage of the context window and instructs the LLM to produce a structured XML snapshot. Other agents combine LLM-based summarization with a set of other context compaction strategies. Claude Code uses a tiered eviction strategy that tries cheap alternatives such as inline compression of large tool outputs before invoking LLM-based summarization~\citep{anthropic2025compactioncookbook, claude2026compaction}. OpenCode~\citep{opencode2026} can additionally prune the outputs of old tool calls. This pruning is separate from its LLM-based summarization, which OpenCode invokes when the conversation exceeds the usable context window.

A large number of context compaction algorithms also use some form of selection instead of LLM-based summarization. The OpenAI Assistants API~\citep{openai2024assistants} offers a truncation strategy parameter with a \texttt{last\_messages} mode that keeps the $N$ most recent messages and an \texttt{auto} mode that drops middle messages to fit a token budget. LangChain's \texttt{trim\_messages} utility~\citep{langchain2024trim} lets users manage context by dropping messages from a conversation according to a configurable strategy, e.g., dropping the oldest messages. Aider's repo map~\citep{gauthier2026aider} ranks code symbols by a PageRank-style importance score and selects the top-$k$ that fit a budget. Token-level prompt compressors such as LLMLingua~\citep{jiang2023llmlingua} and Selective Context~\citep{li2023selective} delete low-information tokens, producing a compressed prompt that is a sub-sequence of the original tokens. 

\section{Models for Context Compaction}
\label{sec:model}

In this section, we formalize context compaction as a two-player game between the context compaction algorithm and an adversary that represents the
future information needs of an agent. We define two games. They differ in the kind of context compaction they model. Our formalization treats the agent's
context as a set of discrete \emph{items} instead of a sequence of tokens.
The items are atomic units of information such as variable names, error messages,
design decisions, or instructions. This abstraction lets us reason about what
information is preserved or lost, independent of how the underlying LLM processes tokens.
Both games model a single invocation of context compaction. Section~\ref{subsec:game} introduces the Context Selection Game, which
captures context compaction algorithms that select a subset of an agent's context.
Section~\ref{subsec:generation} introduces context compaction by generation as
a superclass of context compaction by selection and formalizes the Context Generation
Game. Finally, Section~\ref{subsec:classification} classifies existing
context compaction algorithms used in production agents within our framework.

\subsection{Context Selection Game}
\label{subsec:game}

The Context Selection Game models context compaction algorithms that select a subset of an agent's context. The game is played between a selector and an adversary. A
selector receives an item universe $\calX$ chosen by the adversary. In real applications, $\calX$ is extracted from the agent's
conversation history. The selector must then choose a subset $S \subseteq \calX$ to retain. The subset $S$ is subject to a size budget that captures the LLM's context window
limit. The adversary then issues a query $q$ from a given query space $\calQ$,
and the selector's score is determined by how well the retained subset
$S$ answers $q$ under a value function $v_q$. The formal definition
follows.

\begin{definition}[Context Selection Game]
\label{def:game}
An instance of the Context Selection Game consists of an item universe $\calX$, a size function $s$, a budget $B$, a query space $\calQ$, and a query regime.
The item universe $\calX = \{x_1, \ldots, x_N\}$ represents atomic units of information extracted from a conversation history.
Each item $x_i$ has a size $s(x_i) > 0$ representing the amount of context window space it occupies.
The budget $B > 0$ is the maximum total size of items that can be retained within the context window of an LLM.
Each query $q \in \calQ$ has a value function $v_q: 2^{\calX} \to [0, 1]$ specifying how well a retained subset answers $q$.
The query regime specifies how queries are chosen (see Definition~\ref{def:regimes}).

The game is played between two players, a selector, i.e., the context compaction algorithm, and an adversary, which represents the future information needs of an agent. The game proceeds in three stages. First, the adversary fixes the item universe $\calX = \{x_1, \ldots, x_N\}$ and the size of each item. Second, the selector observes $\calX$ and chooses a subset $S \subseteq \calX$ with $\sum_{x_i \in S} s(x_i) \leq B$. Third, a query $q \in \calQ$ is chosen according to the query regime. The selector is then awarded a score $\val(S, q) := v_q(S)$ and incurs error $\err(S, q) := 1 - \val(S, q)$, which it seeks to minimize. Note that the budget $B$, the query space $\calQ$, and the query regime are fixed before the game starts and are known to both players.
\end{definition}

The difficulty of context compaction depends on what the selector knows about future queries. 
If query patterns are predictable, e.g., users may tend to ask about recent instructions, 
the selector can prioritize accordingly. If queries are adversarial, such as a
red-teamer probing for information the agent forgot, the selector must hedge.
We formalize this as two query regimes.

\begin{definition}[Query regimes]
\label{def:regimes}
We study two query regimes. In the stochastic query regime, the item universe and the query are drawn together from a known distribution $\mu$ over pairs $(\calX, q)$. This draw replaces the adversarial choice of the item universe in the first stage. The selector observes $\calX$ but not $q$. Its error is the expected error under $\mu$. In the oblivious adversary query regime, the adversary fixes the item universe as in the first stage and then chooses the query $q$ with no access to the subset $S$ chosen by the selector. Its error is the worst case of its error over all inputs $(\calX, q)$.
\end{definition}

Note that the stochastic query regime and the oblivious adversary query regime differ in the quantity we are trying to find bounds for. A result that assumes the stochastic regime guarantees a bound on the expected error, where the expectation is over both the draw of $(\calX, q)$ from $\mu$ and the selector's internal randomness. On the other hand, a result that assumes the oblivious adversary regime bounds the error on every input, where the expectation is over the selector's internal randomness alone. The minimum budget in the oblivious adversary regime, at any target error, is therefore at least the minimum budget in the stochastic regime at the same target error, for every distribution $\mu$. For example, a context compaction algorithm designed for a distribution $\mu$ may achieve a low error in expectation and still have a high error on an input that $\mu$ draws rarely. An oblivious adversary can choose that particular bad input.

The choice of what counts as an item is a modeling decision about selection granularity. At the coarsest granularity, a single tool result or the response of a single LLM inference, is one item. At the finest granularity, every token in the internal state of an agent is its own item. Our framework is agnostic to the choice of granularity. Given any partition of the internal state $\calI_t$ into items, the Context Selection Game asks which subset of those items to retain. 
A separate, orthogonal problem to context selection is lossless compression, i.e., encoding an item in fewer tokens without discarding any of its information. Lossless compression applies to each item on its own. Throughout this paper we treat the encoding of an item into tokens as a black box. We measure the size of an item $x_i$ by the size function $s(x_i)$, which accounts for any lossless compression of that item. Sizes are additive across the retained set. The retained set therefore has size $\sum_{x_i \in S} s(x_i)$.

The subset view defines a class of context compaction algorithms. A context compaction algorithm is a \emph{selection algorithm} if its output identifies an \emph{unordered} subset $S \subseteq \calX$. Its cost is the additive $\sum_{x_i \in S} s(x_i)$, where each retained item may be losslessly compressed on its own. We denote the class of selection algorithms $\mathsf{SELECT}$. The Context Selection Game models $\mathsf{SELECT}$.

\subsection{Context Generation Game}
\label{subsec:generation}

We now extend the Context Selection Game to capture context compaction by generation. Our definition of a selection algorithm in Section~\ref{subsec:game} imposes three restrictions. A selection algorithm 1) must output a subset $S \subseteq \calX$, 2) must not use the order of the retained items to convey information, and 3) has an output of size $\sum_{x_i \in S} s(x_i)$, i.e., the sum of the sizes of the retained items. A \emph{generation algorithm} drops all three restrictions. A generation algorithm 1) may emit new tokens that appear in no item of $\calX$, such as the tokens of an LLM-generated summary, 2) may reorder any retained items to convey information, and 3) may encode the whole retained set at a size below $\sum_{x_i \in S} s(x_i)$. For example, a generation algorithm can record which items of $\calX$ belong to $S$ with one bit per item of $\calX$. When $S$ contains most of $\calX$, this encoding is far smaller than $\sum_{x_i \in S} s(x_i)$. We use $\mathsf{GEN}$ to denote the class of generation algorithms. By construction $\mathsf{SELECT} \subseteq \mathsf{GEN}$. LLM-based summarization, used by Codex, Claude Code, and Gemini CLI (Section~\ref{subsec:existing-compaction}), is in $\mathsf{GEN}$ but not necessarily in $\mathsf{SELECT}$. The Context Generation Game parallels the Context Selection Game of Definition~\ref{def:game}. A generator plays against the adversary in place of a selector. The generator's strategy is a condenser--interpreter pair in place of a subset.

\begin{definition}[Context Generation Game]
\label{def:gen-game}
An instance of the Context Generation Game consists of an item universe $\calX$, a budget $B$, a token alphabet $\Sigma$, a query space $\calQ$, and a query regime. Each query $q \in \calQ$ has a value function $v_q : \mathrm{Out} \to [0, 1]$ over an output space $\mathrm{Out}$, specifying how well a candidate output answers $q$.

The game is played between two players, a generator, i.e., the context compaction algorithm together with the LLM that reads the compacted context, and an adversary, as in Definition~\ref{def:game}. The game proceeds in three stages. First, the adversary fixes the item universe $\calX$. Second, the generator applies a \emph{condenser} $\mathsf{Cond}$ that maps the item universe to a message $\mathsf{Cond}(\calX) \in \Sigma^{\leq B}$ of length at most $B$, and commits an \emph{interpreter} $\mathsf{Int}: \Sigma^{\leq B} \times \calQ \to \mathrm{Out}$. Third, a query $q \in \calQ$ is chosen according to the query regime. The generator is then awarded a score $\val(\hat a, q) := v_q(\hat a)$, where $\hat a = \mathsf{Int}(\mathsf{Cond}(\calX), q)$ is the interpreter's output.
\end{definition}

In the case of LLM-based summarization, the condenser is the summarization call and the interpreter is the future LLM call that reads the compacted context. The condenser takes the existing context as input and gives a compacted context or summary as output. The interpreter takes the summary and a query as input, and gives the answer to that query as output.

The Context Selection Game is the special case of the Context Generation Game with the following parameters. The output space is $\mathrm{Out} = 2^{\calX}$. The condenser--interpreter pair lies in $\mathsf{SELECT}$, i.e., $\mathsf{Cond}$'s output identifies a subset $S \subseteq \calX$, possibly after lossless decoding, and $\mathsf{Int}$ returns $S$. The value function $v_q$ scores how well the retained subset $S$ answers $q$. The Context Generation Game removes these restrictions. Any condenser--interpreter pair in $\mathsf{GEN}$ is admissible.

\begin{table}[t]
\centering
\small
\caption{Classification of existing context compaction algorithms by class and, for $\mathsf{SELECT}$, by granularity. Hybrid entries ($\mathsf{SELECT}+\mathsf{GEN}$) stack a $\mathsf{SELECT}$ layer with a $\mathsf{GEN}$ summarization fallback when the $\mathsf{SELECT}$ budget is exhausted; the listed granularity is that of the $\mathsf{SELECT}$ layer.}
\label{tab:classification}
\begin{tabular}{@{}lll@{}}
\toprule
\textbf{System} & \textbf{Class} & \textbf{Granularity} \\
\midrule
Codex~\citep{openai2025codex} & $\mathsf{GEN}$ & --- \\
Gemini CLI~\citep{google2025geminicli} & $\mathsf{GEN}$ & --- \\
\addlinespace
Claude Code~\citep{anthropic2025claudecode} & $\mathsf{SELECT}+\mathsf{GEN}$ & Tool result \\
OpenCode~\citep{opencode2026} & $\mathsf{SELECT}+\mathsf{GEN}$ & Tool result \\
\addlinespace
OpenAI Assistants Truncation Strategy~\citep{openai2024assistants} & $\mathsf{SELECT}$ & Message \\
LangChain \texttt{trim\_messages}~\citep{langchain2024trim} & $\mathsf{SELECT}$ & Message \\
Aider repo map~\citep{gauthier2026aider} & $\mathsf{SELECT}$ & Code symbol \\
LLMLingua~\citep{jiang2023llmlingua} & $\mathsf{SELECT}$ & Token \\
Selective Context~\citep{li2023selective} & $\mathsf{SELECT}$ & Token \\
\bottomrule
\end{tabular}
\end{table}

\subsection{Classifying Existing Context Compaction Mechanisms}
\label{subsec:classification}

Table~\ref{tab:classification} maps the production and research context compaction algorithms surveyed in Section~\ref{subsec:existing-compaction} to the classes $\mathsf{SELECT}$ and $\mathsf{GEN}$ defined in Section~\ref{subsec:generation}, and for $\mathsf{SELECT}$ algorithms records the granularity at which selection operates (Section~\ref{subsec:game}). Every algorithm in our survey lies in $\mathsf{SELECT} \cup \mathsf{GEN}$, and the $\mathsf{SELECT}$ entries span granularities from individual tokens up to entire messages. Our framework therefore covers the vast majority of context compaction algorithms used in production today.

\section{Context Compaction and Communication}
\label{sec:communication}

In this section, we observe that the Context Generation Game of
Definition~\ref{def:gen-game} can be cast directly as an instance of
one-way communication, with the condenser playing Alice and the interpreter
playing Bob. Section~\ref{subsec:induced-comm} makes this correspondence
formal: it defines the one-way communication problem induced by a Context
Generation Game and proves the two equivalent.
Section~\ref{subsec:select-as-restricted} specializes the equivalence to
the $\mathsf{SELECT}$ subclass, identifying $\mathsf{SELECT}$ algorithms as
one-way protocols of a restricted form. It also exhibits a family of queries on
which selection needs a factor of $\Theta(\log n)$ more budget than generation.
Section~\ref{subsec:computation-vs-communication} records two caveats about
the scope of the equivalence: that it is a statement about information rather
than computation, and that it is restricted to the stochastic and oblivious
regimes.

\subsection{The Induced Communication Problem}
\label{subsec:induced-comm}

We make the correspondence formal by associating each Context Generation
Game with a one-way communication problem, defined as follows.

\begin{definition}[Generation $\Rightarrow$ Communication]
\label{def:induced-comm}
Let $\mathcal{G}$ be an instance of a Context Generation Game
(Definition~\ref{def:gen-game}), with item universe $\calX$, budget $B$,
token alphabet $\Sigma$, query space $\calQ$, value functions
$\{v_q\}_{q \in \calQ}$, and output space $\mathrm{Out}$. We define the
one-way communication problem
induced by $\mathcal{G}$, $\Pi_\mathcal{G}$, as follows. $\Pi_\mathcal{G}$ 
is the two-party task in which Alice receives the item universe $\calX$
and Bob receives a query $q \in \calQ$. Alice then sends a single message 
$m \in \Sigma^{\leq B}$ to Bob. Bob then outputs an answer $\hat a \in \mathrm{Out}$
for the query $q$. The protocol's error on input $(\calX, q)$ is defined as
$1 - v_q(\hat a)$ where $v_q$ is the value function associated with query $q$.
\end{definition}

For any input distribution $\mu$ over $\calX \times \calQ$, we 
define $\mathsf{R}^{\to}_{\mu,\epsilon}(\Pi)$ to be the 
$\mu$-distributional one-way public-coin randomized 
communication complexity of $\Pi$ at error $\epsilon$. 
In other words, it is the minimum message length in bits of a protocol whose expected error under
$\mu$ is at most $\epsilon$. 
We define $\mathsf{R}^{\to}_\epsilon(\Pi)$ to be the minimum message length
of a public-coin randomized protocol whose expected error, over its public
coins, is at most $\epsilon$ on \emph{every} input $(\calX, q)$.
Note that in this section we
measure the budget for the Context Generation Game in bits and
identify the token alphabet $\Sigma = \{0, 1\}$. For $|\Sigma| > 1$, 
$B$ symbols are equivalent to $B \log|\Sigma|$ bits.

We are now ready to formalize the equivalence between generation and one-way communication.

\begin{theorem}[Generation $\equiv$ one-way communication]
\label{thm:gen-equiv}
Let $\mathcal{G}$ be a Context Generation Game and $\Pi_\mathcal{G}$ be the
one-way communication problem induced by $\mathcal{G}$ as in Definition~\ref{def:induced-comm}.
The following properties hold.

\textbf{Property 1.} For any input distribution $\mu$ over pairs $(\calX, q)$, the minimum
budget at which some (possibly randomized) $\mathsf{GEN}$ algorithm achieves
expected error $\leq \epsilon$ with the input drawn from $\mu$ equals
$\mathsf{R}^{\to}_{\mu,\epsilon}(\Pi_\mathcal{G})$. The expectation is over both
the algorithm's internal randomness and the input $(\calX, q) \sim \mu$. Note that this is
the stochastic query regime, with the conversation and the query drawn jointly
from $\mu$.

\textbf{Property 2.} In the oblivious adversary query regime, the minimum budget at which
some (possibly randomized) $\mathsf{GEN}$ algorithm achieves expected error
$\leq \epsilon$ on \emph{every} input $(\calX, q)$ equals
$\mathsf{R}^{\to}_\epsilon(\Pi_\mathcal{G})$. The expectation is over the
algorithm's internal randomness, and the bound is required on every input,
so the controlled quantity is the worst case over inputs of the expected
error.
\end{theorem}

\begin{proof}
The proof rests on the observation that a $\mathsf{GEN}$ algorithm $G$ for a Context Generation Game $\mathcal{G}$ and a one-way
communication protocol $\pi$ for the induced one-way communication problem $\Pi_\mathcal{G}$ are described by the same pair of
functions. The functions are 1) a condenser $\mathsf{Cond}$, which Alice runs to turn the item universe into a
message, and 2) an interpreter $\mathsf{Int}$, which Bob runs to turn a message and query into an
answer.

A deterministic $\mathsf{GEN}$ algorithm $G$ is a condenser $\mathsf{Cond}$ and an interpreter $\mathsf{Int}$.
The condenser turns the item universe $\calX$ into a message $\mathsf{Cond}(\calX)$ of at
most $B$ bits, and the interpreter turns that message and a query $q$ into an
answer. This is identical to a one-way communication protocol $\pi$. In $\pi$, Alice is the condenser. Alice holds $\calX$ and sends the message $\mathsf{Cond}(\calX)$. Bob is the interpreter. Bob holds $q$, receives the message $m$, and outputs $\mathsf{Int}(m, q)$. The answer
$\mathsf{Int}(\mathsf{Cond}(\calX), q)$ is identical in both $G$ and $\pi$. Therefore, the error $1 - v_q$ on each input
$(\calX, q)$ is equal. So every $\mathsf{GEN}$ algorithm $G$ of budget $B$ has a
corresponding one-way communication protocol $\pi$ of message length $B$ with the same error on every input
$(\calX, q)$, namely the protocol in which Alice runs the condenser $\mathsf{Cond}$ and Bob
runs the interpreter $\mathsf{Int}$. Conversely, every one-way communication protocol $\pi$ of message length $B$
has a corresponding $\mathsf{GEN}$ algorithm $G$ of budget $B$ with the same error on every
input, obtained by taking the condenser to be Alice's strategy and the interpreter
to be Bob's.

A randomized $\mathsf{GEN}$ algorithm $G_r$ is a distribution over deterministic
pairs $(\mathsf{Cond}_\rho, \mathsf{Int}_\rho)$, indexed by its internal coins $\rho$, since fixing the
coins makes the condenser and interpreter deterministic. Likewise a public-coin
one-way communication protocol $\pi_r$ is, by definition, a deterministic
protocol for each setting of its shared public coins. We identify $G_r$ and
$\pi_r$ by drawing both from the same coins $\rho$. For each outcome $\rho$, the
pair $(\mathsf{Cond}_\rho, \mathsf{Int}_\rho)$ is, by the deterministic case above, at once a
$\mathsf{GEN}$ algorithm of budget $B$ and a one-way protocol of message length
$B$ that make the same error on every input. Averaging over $\rho$, $G_r$ and
$\pi_r$ then have the same expected error on every input. So every randomized
$\mathsf{GEN}$ algorithm $G_r$ of budget $B$ has a public-coin protocol $\pi_r$
of message length $B$ with the same expected error on every input, and
conversely.

Both properties now follow by reading this correspondence under two different
error measures, with no further work. For property 1, the error is the expected error under $\mu$. As shown above, $G_r$ and $\pi_r$ have the same expected error on every input, so the
same expected error under $\mu$. Hence $G_r$ achieves expected error
$\leq \epsilon$ at budget $B$ if and only if $\pi_r$ does at message length $B$,
and the smallest $\mathsf{GEN}$ budget equals
$\mathsf{R}^{\to}_{\mu,\epsilon}(\Pi_\mathcal{G})$. For property 2, the error is instead the worst case over inputs of
the expected error. Since $G_r$ and $\pi_r$ have the same expected error on
every input, they have the same worst case over inputs, so the smallest
$\mathsf{GEN}$ budget for worst-case error $\leq \epsilon$ equals
$\mathsf{R}^{\to}_\epsilon(\Pi_\mathcal{G})$.
\end{proof}

\subsection{Selection as Restricted Communication}
\label{subsec:select-as-restricted}

The equivalence of Theorem~\ref{thm:gen-equiv} is for the unrestricted class
$\mathsf{GEN}$. Since $\mathsf{SELECT}$ is a subclass of $\mathsf{GEN}$, the
equivalence specializes to a restricted family of one-way protocols. The
following corollary makes the restriction explicit and identifies any gap
between $\mathsf{SELECT}$ and $\mathsf{GEN}$ as a gap within the protocol
class.

\begin{corollary}
\label{cor:select-restricted}
Let $\calX$ be an item universe of size $N$, $B$ be a budget in bits, and $s(x_i)$ be the size
function for item $x_i \in \calX$. $\mathsf{SELECT}$ algorithms correspond exactly to one-way protocols in which
Alice's message identifies a subset $S \subseteq [N]$ such that
$\sum_{i \in S} s(x_i) \leq B$, and Bob's output is a function only of
$S$, the set $\{x_i : i \in S\}$, $q$, and the public coins. In particular, every
gap between $\mathsf{SELECT}$ and $\mathsf{GEN}$ on a set of queries $\calQ$ is
a gap between subset-encoding protocols and unrestricted one-way protocols
on $\Pi_{\mathcal{G}_\calQ}$.
\end{corollary}

\begin{proof}
We first show that any $\mathsf{SELECT}$ algorithm with budget $B$
corresponds to a one-way communication protocol of the restricted form described in the
corollary. Suppose $(\mathsf{Cond}, \mathsf{Int})$ is a $\mathsf{SELECT}$ algorithm with budget
$B$. By the definition of $\mathsf{SELECT}$ in
Section~\ref{subsec:generation}, the condenser $\mathsf{Cond}$ produces an output
that, after lossless decoding, identifies a subset $S \subseteq [N]$
satisfying $\sum_{i \in S} s(x_i) \leq B$. The interpreter $\mathsf{Int}$ produces an
answer that is a function only of $S$, $\{x_i : i \in S\}$, and the query
$q$. Following the construction in the proof of
Theorem~\ref{thm:gen-equiv}, we define a one-way communication protocol in which Alice,
on input $\calX$, sends $\mathsf{Cond}(\calX)$ to Bob, and Bob, on input $q$ and
message $m$, outputs $\mathsf{Int}(m, q)$. By construction, Alice's message
identifies the subset $S$ chosen by $\mathsf{Cond}(\calX)$, and Bob's output is a
function only of $S$, $\{x_i : i \in S\}$, and $q$. Conversely, suppose $\pi$ is a one-way communication protocol whose message identifies a
subset $S \subseteq [N]$ with $\sum_{i \in S} s(x_i) \leq B$, and whose
interpreter output is a function only of $S$, $\{x_i : i \in S\}$, $q$, and
the public coins. We set $\mathsf{Cond}$ to be Alice's strategy and $\mathsf{Int}$ to be Bob's
strategy. The condenser $\mathsf{Cond}$'s output identifies a subset of $\calX$ within
the budget, and the interpreter $\mathsf{Int}$ is a function only of that subset and the
query. Therefore, the condenser--interpreter pair $(\mathsf{Cond}, \mathsf{Int})$ is a $\mathsf{SELECT}$ algorithm with budget
$B$.

The gap claim follows by applying Theorem~\ref{thm:gen-equiv} to both
classes. The minimum $\mathsf{SELECT}$ budget achieving error
$\leq \epsilon$ equals the minimum length of a subset-encoding one-way
communication protocol at error $\leq \epsilon$. Similarly, the minimum $\mathsf{GEN}$ budget
at error $\leq \epsilon$ equals the minimum length of an unrestricted
one-way communication protocol at error $\leq \epsilon$. Any gap between the two minima
is therefore a gap between subset-encoding protocols and unrestricted
one-way communication protocols on $\Pi_{\mathcal{G}_\calQ}$.
\end{proof}

We now exhibit a query on which $\mathsf{SELECT}$ algorithms need strictly more budget than $\mathsf{GEN}$ algorithms.

\begin{theorem}[Selection $\subsetneq$ Generation]
\label{prop:select-gen-gap}
Let $n = 2^k$ for an integer $k \geq 2$. There is a set $Y$ of $n$ items and a single query $q$ on $Y$ with the following property: some $\mathsf{GEN}$ algorithm answers $q$ with zero error using a budget of $n$ bits, while every $\mathsf{SELECT}$ algorithm that answers $q$ with zero error requires a budget of at least $nk = n \log_2 n > n$ bits.
\end{theorem}

Note we choose $n$ a power of two for notational convenience.

\begin{proof}
Let each item of $Y$ have size $s(x) = \log_2 n$ bits. The adversary fixes the item universe $\calX$ to be an arbitrary, possibly empty, subset of $Y$. The query $q$ asks for $\calX$: an answer is correct if and only if it equals $\calX$.

For the upper bound, the condenser sends the $n$-bit indicator vector of $\calX$ over $Y$. The interpreter returns the set that the vector indicates. This $\mathsf{GEN}$ algorithm has zero error with a budget of $n$ bits.

The intuition for the lower bound is that the full set of items $Y$ does not fit within a budget below $n \log_2 n$. Therefore, the set of choices a $\mathsf{SELECT}$ algorithm can make will not span all subsets of $Y$. Thus, a $\mathsf{SELECT}$ algorithm cannot perfectly encode an arbitrary subset of $Y$ within the budget. More formally, let $A_{S}$ be a deterministic $\mathsf{SELECT}$ algorithm with zero error and budget $B$. By Corollary~\ref{cor:select-restricted}, on input $\calX$, $A_{S}$ retains a subset $S(\calX) \subseteq \calX$ of total size at most $B$. Since the query is fixed, the answer of $A_{S}$ is $f(S(\calX))$ for a fixed function $f$. The zero error assumption ensures $f(S(\calX)) = \calX$ for every $\calX \subseteq Y$. The map $\calX \mapsto S(\calX)$ is therefore an injection of the power set of $Y$ into itself, hence a bijection. In particular, $S(\calX^\star) = Y$ for some $\calX^\star$. Since $S(\calX^\star) \subseteq \calX^\star$, we have $\calX^\star = Y$. On input $Y$, $A_{S}$ retains all $n$ items at a cost of $n \log_2 n$ bits. Hence $B \geq n \log_2 n$. If $A_{S}$ is a randomized algorithm, the error must be zero for every value of its public random coins, so the bound still holds for randomized $\mathsf{SELECT}$ algorithms.
\end{proof}

\subsection{Computation Caveats}
\label{subsec:computation-vs-communication}

In this section, we note two caveats regarding
the scope of Theorem~\ref{thm:gen-equiv}. The first caveat concerns the
distinction between information and computation. The second
caveat identifies the query regimes for which the equivalence in Theorem~\ref{thm:gen-equiv} holds.

The first caveat is that the equivalence of Theorem~\ref{thm:gen-equiv} is
information-theoretic. Our bounds are on the context compaction budget, i.e., the
number of bits the compacted context must use to answer a set of queries within
error $\epsilon$. Our bounds concern the information needed for the context compaction but do not address the
computation a context compaction algorithm must perform to attain it. An upper bound on the context compaction budget
exhibits some $\mathsf{GEN}$ algorithm that meets it, but does not guarantee
that a particular context compaction algorithm, such as a summarizer whose interpreter is an LLM call,
can compute it. A lower bound on the context compaction budget does not have this caveat. The lower bound says
the compacted context must carry at least a given number of bits, which no interpreter,
however it is computed, can reduce. By Theorem~\ref{thm:gen-equiv},
the minimum $\mathsf{GEN}$ budget equals the minimum one-way protocol budget for
the induced problem $\Pi_\mathcal{G}$, so a lower bound for $\Pi_\mathcal{G}$
bounds every $\mathsf{GEN}$ algorithm, regardless of how its interpreter is
computed.

The second caveat is that the equivalence of Theorem~\ref{thm:gen-equiv} is also restricted to
the stochastic and oblivious query regimes of Definition~\ref{def:regimes}.
Under an \emph{adaptive} adversary, which observes the condenser's output
before choosing $q$, Theorem~\ref{thm:gen-equiv} does not apply. We leave
characterizing the adaptive case as an open problem in
Section~\ref{sec:discussion}.

\section{Discussion and Open Problems}
\label{sec:discussion}

Theorem~\ref{thm:gen-equiv} helps an agent designer calculate the minimum context compaction budget required for a given set of possible future queries or a distribution over future queries that a user might ask. The minimum context compaction budget for answering these queries within a target error is equal to the one-way communication complexity of the induced communication problem at the same error. Prior work in the field of communication complexity has computed the one-way communication complexity for many types of queries such as set membership queries (see Appendix~\ref{app:empirical}) and equality queries~\citep{kushilevitz1997communication}.

The theorem also helps agent designers understand the fundamental limitations of context compaction. For example, consider an agent that scans a repository and records the set of third-party dependencies it uses. Assume the agent must answer the following query.

\begin{querybox}
Do any of these dependencies appear in this list of packages with known Common Vulnerabilities and Exposures (CVEs)?
\end{querybox}

This is a set disjointness query, where $S$ is the set of dependencies and $T$ is the set of packages with known CVEs, and the user is asking whether $S \cap T = \emptyset$. Set disjointness is a canonical hard problem in communication complexity~\citep{roughgarden2016communication}. Suppose each package is identified by an $m$-bit string, so there are $|U| = 2^m$ possible packages, and let $N = |S|$ be the number of dependencies. Note that this set disjointness query contains set membership as a special case. For a single package $T = \{x\}$, answering $S \cap T = \emptyset$ is the same as answering $x \in S$. Since the compacted context is fixed before the query arrives, the LLM must be able to answer this membership query for any package $x \in U$. Therefore, to answer every such query correctly, the compacted context must determine $S$ exactly. Representing an $N$-element subset of $U$ requires $\log_2 \binom{|U|}{N} \geq N \log_2(|U|/N)$ bits~\citep{carter1978exact}. The number of dependencies is far smaller than the number of possible packages, so this is $\Omega(N m)$ bits. By Theorem~\ref{thm:gen-equiv}, every strategy for the Context Generation Game that answers this query correctly on all inputs must use a context compaction budget of $\Omega(N m)$ bits. This is no better than storing the dependencies uncompressed, which may be infeasible for large repositories.

One might hope to sidestep this by answering the query with approximate membership queries, e.g., using a Bloom filter. We could query each of the packages in $T$ for membership in $S$ and report an intersection if any query succeeds. This approach does not help significantly. A Bloom filter over the $N$ recorded dependencies with false positive rate $\epsilon$ uses about $1.44 \, N \log_2(1/\epsilon)$ bits. When $S$ and $T$ are disjoint, every package in $T$ is a true non-member, so the query reports a spurious intersection unless all $|T|$ membership queries answer in the negative. For an idealized Bloom filter, false positives on distinct items occur independently. All membership queries are therefore negative with probability $(1-\epsilon)^{|T|}$, and the set disjointness query errs with probability $1 - (1-\epsilon)^{|T|}$. Keeping this error below a target $\delta$ requires $(1-\epsilon)^{|T|} \ge 1-\delta$, and since $(1-\epsilon)^{|T|} \le e^{-\epsilon |T|}$, this forces $\epsilon\,|T| \le \ln\frac{1}{1-\delta} = O(\delta)$. Hence $\epsilon = O(\delta/|T|)$, and the filter must use $\Omega(N \log_2(|T|/\delta))$ bits. Since the query may name any package, $|T|$ can be as large as $|U| = 2^m$, so this is again $\Omega(N m)$ bits, the same as storing the dependencies outright. This holds even when the target error $\delta$ is a constant. This example illustrates how Theorem~\ref{thm:gen-equiv} can inform the design and limitations of context compaction in practice.

Our work also raises open problems. We state four of them below.

\paragraph{Adaptive adversaries.}
Theorem~\ref{thm:gen-equiv} covers the stochastic and oblivious query regimes of Definition~\ref{def:regimes}. An adaptive adversary observes the condenser's output before choosing its query. Theorem~\ref{thm:gen-equiv} does not apply in this regime. We leave finding a communication model whose complexity equals the minimum context compaction budget under an adaptive adversary as an open problem.

\paragraph{Separating selection from generation.}
A set of queries admits a gap between $\mathsf{SELECT}$ and $\mathsf{GEN}$ when the minimum $\mathsf{SELECT}$ budget at a given error exceeds the minimum $\mathsf{GEN}$ budget at the same error. Corollary~\ref{cor:select-restricted} shows that any such gap is a gap between subset-encoding protocols and unrestricted one-way protocols. Theorem~\ref{prop:select-gen-gap} exhibits one such gap, of size $\Theta(\log n)$. Characterizing which sets of queries admit a gap, and how large the gap can grow in general, remains open. Such a characterization would tell agent designers when generation-based context compaction can outperform selection-based context compaction.

\paragraph{Repeated context compaction.}
Our model captures a single invocation of context compaction. Agents do context compaction many times over a long session. In deployed agents, the compacted summary typically replaces the conversation history that the agent loop uses, so a later context compaction acts on a state that contains the summaries produced by earlier context compactions. Codex warns its users that repeated context compaction degrades quality, as we discussed in Section~\ref{subsec:existing-compaction}. We leave extending our model to a sequence of context compactions as an open problem. The goal is to characterize how the error on a fixed set of queries grows with the number of context compactions.

\paragraph{Computationally efficient context compaction.}
Our bounds are information-theoretic, as we discussed in Section~\ref{subsec:computation-vs-communication}. A budget that is attainable in principle may not be attainable by a realistic context compaction algorithm such as an LLM-based summarizer. Determining which optimal budgets remain attainable when the condenser and the interpreter must run in polynomial time, or when the interpreter is an LLM call, is an open problem. A concrete first question is whether an LLM can reliably simulate the decoding procedure of a sketch~\citep{cormode2005improved} placed in its context, as we discuss in Appendix~\ref{app:empirical}.

\section*{Generative AI Disclosure}

The vast majority of the content of this paper is human written. Generative AI was used as an assistant for the following three purposes. The first purpose was to help the authors review and make targeted corrections to the text that improve grammar, spelling, mathematical notation, correctness of the proofs, \LaTeX{} code, and clarity.
The second purpose was to help the authors write the code for the empirical case study in Appendix~\ref{app:empirical}, which was then reviewed and corrected by the authors. The final purpose was to help the authors search for related work, which was then read and verified by the authors. Note that all of the content of this paper has been human reviewed and verified for correctness prior to being declared 
ready for the purposes of a preprint. The following LLMs were used for the above discussed purposes: Anthropic's Claude Opus 4.8 and Claude Fable 5, DeepSeek's DeepSeek R1, and Google's Gemma 4 31B Thinking.

\bibliographystyle{plainnat}
\bibliography{references}

@misc{openai2025codex,
  author       = {{OpenAI}},
  title        = {{Codex CLI}: Lightweight Coding Agent that Runs in Your Terminal},
  year         = {2025},
  howpublished = {\url{https://github.com/openai/codex}},
  note         = {Accessed: 2026-04-25}
}

@misc{opencode2026,
  author       = {{SST}},
  title        = {{OpenCode}: The Open Source AI Coding Agent},
  year         = {2026},
  howpublished = {\url{https://github.com/sst/opencode}},
  note         = {Accessed: 2026-04-25}
}

@misc{block2025goose,
  author       = {{Block}},
  title        = {{Goose}: An Open-Source, Extensible AI Agent},
  year         = {2025},
  howpublished = {\url{https://github.com/aaif-goose/goose}},
  note         = {Accessed: 2026-04-25}
}

@misc{ibmreact,
  howpublished = {\url{https://www.ibm.com/think/topics/react-agent}},
  title        = {ReAct Agent},
  author = {{Dave Bergmann}},
  publisher       = {{IBM}},
  accessed      = {2026-04-29},
  year         = {2026}
}

@inproceedings{yao2022react,
  title={ReAct: Synergizing Reasoning and Acting in Language Models},
  author={Yao, Shunyu and Zhao, Jeffrey and Yu, Dian and Du, Nan and Shafran, Izhak and Narasimhan, Karthik and Cao, Yuan},
  booktitle={International Conference on Learning Representations (ICLR)},
  year={2023}
}

@inproceedings{almashaqbeh2025adversary,
  title={Adversary resilient learned bloom filters},
  author={Almashaqbeh, Ghada and Bishop, Allison and Tirmazi, Hayder},
  booktitle={International Conference on the Theory and Application of Cryptology and Information Security},
  pages={171--202},
  year={2025},
  organization={Springer}
}

@article{zhang2025recursive,
  title={Recursive language models},
  author={Zhang, Alex L and Kraska, Tim and Khattab, Omar},
  journal={arXiv preprint arXiv:2512.24601},
  year={2025}
}

@article{hergert2025brittleness,
  title={On the Brittleness of {LLMs}: A Journey around Set Membership},
  author={Hergert, Lea and Berend, G{\'a}bor and Szegedy, Mario and Tur{\'a}n, Gy{\"o}rgy and Jelasity, M{\'a}rk},
  journal={arXiv preprint arXiv:2511.12728},
  year={2025}
}

@article{guo2026hallucination,
  title={Hallucination is a Consequence of Space-Optimality: A Rate-Distortion Theorem for Membership Testing},
  author={Guo, Anxin and Li, Jingwei},
  journal={arXiv preprint arXiv:2602.00906},
  year={2026}
}

@inproceedings{rae2019metalearning,
  title={Meta-Learning Neural Bloom Filters},
  author={Rae, Jack W and Bartunov, Sergey and Lillicrap, Timothy P},
  booktitle={Proceedings of the 36th International Conference on Machine Learning (ICML)},
  year={2019}
}

@article{liu2023lost,
  title={Lost in the Middle: How Language Models Use Long Contexts},
  author={Liu, Nelson F and Lin, Kevin and Hewitt, John and Paranjape, Ashwin and Bevilacqua, Michele and Petroni, Fabio and Liang, Percy},
  journal={Transactions of the Association for Computational Linguistics},
  volume={12},
  pages={157--173},
  year={2024}
}

@article{schmidgall2025agent,
  title={Agent laboratory: Using llm agents as research assistants},
  author={Schmidgall, Samuel and Su, Yusheng and Wang, Ze and Sun, Ximeng and Wu, Jialian and Yu, Xiaodong and Liu, Jiang and Moor, Michael and Liu, Zicheng and Barsoum, Emad},
  journal={Findings of the Association for Computational Linguistics: EMNLP 2025},
  pages={5977--6043},
  year={2025},
  publisher={Association for Computational Linguistics}
}

@misc{anthropic2025compactioncookbook,
    author       = {{Anthropic}},
    title        = {Automatic Context Compaction},
    year         = {2025},
    howpublished = {\url{https://platform.claude.com/cookbook/tool-use-automatic-context-compaction}},
    note         = {Accessed: 2026-04-13}
  }

@article{zhu2026paperbanana,
  title={PaperBanana: Automating Academic Illustration for AI Scientists},
  author={Zhu, Dawei and Meng, Rui and Song, Yale and Wei, Xiyu and Li, Sujian and Pfister, Tomas and Yoon, Jinsung},
  journal={arXiv preprint arXiv:2601.23265},
  year={2026}
}

@article{kang2025acon,
  title={Acon: Optimizing context compression for long-horizon llm agents},
  author={Kang, Minki and Chen, Wei-Ning and Han, Dongge and Inan, Huseyin A and Wutschitz, Lukas and Chen, Yanzhi and Sim, Robert and Rajmohan, Saravan},
  journal={arXiv preprint arXiv:2510.00615},
  year={2025}
}

@misc{anthropic2025multiagent,
    author       = {{Anthropic}},
    title        = {How We Built Our Multi-Agent Research System},
    year         = {2025},
    howpublished = {\url{https://www.anthropic.com/engineering/multi-agent-research-system}},
    note         = {Accessed: 2026-04-13}
  }

@misc{anthropic2025claudecode,
  author       = {{Anthropic}},
  title        = {Claude {C}ode},
  year         = {2025},
  howpublished = {\url{https://code.claude.com/docs}},
  note         = {Accessed: 2026-04-13}
}

@misc{google2025geminicli,
  author       = {{Google}},
  title        = {{Gemini CLI}},
  year         = {2025},
  howpublished = {\url{https://github.com/google-gemini/gemini-cli}},
  note         = {Accessed: 2026-04-13}
}

@misc{cursor2024cursor,
  author       = {{Anysphere}},
  title        = {{Cursor}: The {AI} Code Editor},
  year         = {2024},
  howpublished = {\url{https://cursor.com}},
  note         = {Accessed: 2026-04-13}
}

@misc{stripe2026minions,
  author       = {{Stripe}},
  title        = {Minions: {S}tripe's One-Shot, End-to-End Coding Agents},
  year         = {2026},
  howpublished = {\url{https://stripe.dev/blog/minions-stripes-one-shot-end-to-end-coding-agents}},
  note         = {Accessed: 2026-04-13}
}

@misc{langgraph,
  author = {LangChain},
  title = {LangGraph},
  year = {2026},
  howpublished = {\url{https://www.langchain.com/langgraph}}
}

@article{li2026latent,
  title={Latent Context Compilation: Distilling Long Context into Compact Portable Memory},
  author={Li, Zeju and Zhou, Yizhou and Xu, Qiang},
  journal={arXiv preprint arXiv:2602.21221},
  year={2026}
}

@inproceedings{mitzenmacher2018model,
  title={A Model for Learned {B}loom Filters and Optimizing by Sandwiching},
  author={Mitzenmacher, Michael},
  booktitle={Advances in Neural Information Processing Systems},
  volume={31},
  year={2018}
}

@inproceedings{kraska2018case,
  title={The case for learned index structures},
  author={Kraska, Tim and Beutel, Alex and Chi, Ed H and Dean, Jeffrey and Polyzotis, Neoklis},
  booktitle={Proceedings of the 2018 International Conference on Management of Data},
  pages={489--504},
  year={2018}
}

@inproceedings{jiang2023llmlingua,
  title={{LLMLingua}: Compressing Prompts for Accelerated Inference of Large Language Models},
  author={Jiang, Huiqiang and Wu, Qianhui and Lin, Chin-Yew and Yang, Yuqing and Qiu, Lili},
  booktitle={Proceedings of the 2023 Conference on Empirical Methods in Natural Language Processing},
  year={2023}
}

@inproceedings{jiang2024llmlingua2,
  title={{LLMLingua-2}: Data Distillation for Efficient and Faithful Task-Agnostic Prompt Compression},
  author={Jiang, Huiqiang and Wu, Qianhui and others},
  booktitle={Findings of the Association for Computational Linguistics: ACL 2024},
  pages={963--981},
  year={2024}
}

@misc{claude2026compaction,
  title={Compaction},
  author={{Anthropic}},
  howpublished={\url{https://platform.claude.com/docs/en/build-with-claude/compaction}},
  year={2026}
}

@article{broder2004network,
  title={Network applications of {Bloom} filters: A survey},
  author={Broder, Andrei and Mitzenmacher, Michael},
  journal={Internet Mathematics},
  volume={1},
  number={4},
  pages={485--509},
  year={2004}
}

@misc{shahout2026orlalibraryservingllmbased,
      title={Orla: A Library for Serving LLM-Based Multi-Agent Systems},
      author={Rana Shahout and Hayder Tirmazi and Minlan Yu and Michael Mitzenmacher},
      year={2026},
      eprint={2603.13605},
      archivePrefix={arXiv},
      primaryClass={cs.AI},
      url={https://arxiv.org/abs/2603.13605},
}

@misc{siddhartha2021malicious,
  title={Malicious {URLs} Dataset},
  author={Siddhartha, Manu},
  howpublished={\url{https://www.kaggle.com/datasets/sid321axn/malicious-urls-dataset}},
  year={2021}
}

@article{bloom1970space,
  title={Space/time trade-offs in hash coding with allowable errors},
  author={Bloom, Burton H},
  journal={Communications of the ACM},
  volume={13},
  number={7},
  pages={422--426},
  year={1970},
  publisher={ACM}
}

@inproceedings{carter1978exact,
  title={Exact and approximate membership testers},
  author={Carter, Larry and Floyd, Robert and Gill, John and Markowsky, George and Wegman, Mark},
  booktitle={Proceedings of the Tenth Annual ACM Symposium on Theory of Computing (STOC)},
  pages={59--65},
  year={1978}
}

@inproceedings{pagh2013approximate,
  title={How to approximate a set without knowing its size in advance},
  author={Pagh, Rasmus and Segev, Gil and Wieder, Udi},
  booktitle={Proceedings of the 54th Annual IEEE Symposium on Foundations of Computer Science (FOCS)},
  pages={80--89},
  year={2013}
}

@misc{openai2024assistants,
  author       = {{OpenAI}},
  title        = {Assistants {API}: Truncation Strategy},
  year         = {2024},
  howpublished = {\url{https://developers.openai.com/api/docs/assistants/deep-dive}},
  note         = {Accessed: 2026-04-26}
}

@misc{awsbedrock,
  author = {{Amazon Web Services}},
  title = {AWS Bedrock},
  year = {2026},
  howpublished = {\url{https://aws.amazon.com/bedrock/}},
  note = {Accessed: 2026-06-23}
}

@misc{langchain2024trim,
  author       = {{LangChain}},
  title        = {Trim Messages},
  year         = {2024},
  howpublished = {\url{https://reference.langchain.com/python/langchain-core/messages/utils/trim_messages}},
  note         = {Accessed: 2026-04-26}
}

@misc{gauthier2026aider,
  author       = {Gauthier, Paul},
  title        = {Aider: {AI} Pair Programming in Your Terminal},
  year         = {2026},
  howpublished = {\url{https://aider.chat/docs/repomap.html}},
  note         = {Accessed: 2026-04-26}
}

@inproceedings{xiao2024streamingllm,
  title     = {Efficient Streaming Language Models with Attention Sinks},
  author    = {Xiao, Guangxuan and Tian, Yuandong and Chen, Beidi and Han, Song and Lewis, Mike},
  booktitle = {International Conference on Learning Representations (ICLR)},
  year      = {2024}
}

@inproceedings{zhang2023h2o,
  title     = {{H2O}: Heavy-Hitter Oracle for Efficient Generative Inference of Large Language Models},
  author    = {Zhang, Zhenyu and Sheng, Ying and Zhou, Tianyi and Chen, Tianlong and Zheng, Lianmin and Cai, Ruisi and Song, Zhao and Tian, Yuandong and R{\'e}, Christopher and Barrett, Clark and Wang, Zhangyang and Chen, Beidi},
  booktitle = {Advances in Neural Information Processing Systems (NeurIPS)},
  year      = {2023}
}

@inproceedings{li2023selective,
  title     = {Compressing Context to Enhance Inference Efficiency of Large Language Models},
  author    = {Li, Yucheng and Dong, Bo and Guerin, Frank and Lin, Chenghua},
  booktitle = {Conference on Empirical Methods in Natural Language Processing (EMNLP)},
  year      = {2023}
}

@book{kushilevitz1997communication,
  title     = {Communication Complexity},
  author    = {Kushilevitz, Eyal and Nisan, Noam},
  year      = {1997},
  publisher = {Cambridge University Press}
}

@article{cormode2005improved,
  title     = {An improved data stream summary: the count-min sketch and its applications},
  author    = {Cormode, Graham and Muthukrishnan, S.},
  journal   = {Journal of Algorithms},
  volume    = {55},
  number    = {1},
  pages     = {58--75},
  year      = {2005},
  publisher = {Elsevier}
}

@misc{anthropic2026pricing,
  title        = {Pricing},
  author       = {{Anthropic}},
  howpublished = {\url{https://platform.claude.com/docs/en/about-claude/pricing}},
  year         = {2026},
  note         = {Accessed 2026-07-20}
}

@misc{artificialanalysis2026fable,
  title        = {Claude Fable 5: Intelligence, Performance and Price Analysis},
  author       = {{Artificial Analysis}},
  howpublished = {\url{https://artificialanalysis.ai/models/claude-fable-5}},
  year         = {2026},
  note         = {Accessed 2026-07-20}
}

@misc{anthropic2026promptcaching,
  title        = {Prompt caching},
  author       = {{Anthropic}},
  howpublished = {\url{https://platform.claude.com/docs/en/build-with-claude/prompt-caching}},
  year         = {2026},
  note         = {Accessed 2026-07-20}
}

@article{roughgarden2016communication,
  title={Communication complexity (for algorithm designers)},
  author={Roughgarden, Tim},
  journal={Foundations and Trends in Theoretical Computer Science},
  volume={11},
  number={3-4},
  pages={217--404},
  year={2016},
  publisher={now Publishers}
}

@inproceedings{adabf,
  author       = {Zhenwei Dai and
                  Anshumali Shrivastava},
  title        = {Adaptive Learned Bloom Filter (Ada-BF): Efficient Utilization of the
                  Classifier with Application to Real-Time Information Filtering on
                  the Web},
  booktitle    = {Advances in Neural Information Processing Systems (NeurIPS)},
  year         = {2020}
}

@inproceedings{sato_matsui,
author = {Sato, Atsuki and Matsui, Yusuke},
title = {Fast partitioned learned bloom filter},
booktitle = {International Conference on Neural Information Processing Systems (NeurIPS)},
year = {2023}
}

@inproceedings{plbf,
  title={Partitioned Learned Bloom Filters},
  author={Vaidya, Kapil and Knorr, Eric and Mitzenmacher, Michael and Kraska, Tim},
  booktitle={International Conference on Learning Representations},
  year= {2021}
}

\appendix

\section{Context Compaction in the Wild}
\label{app:empirical}

We test how close a deployed context compaction endpoint comes to the optimal context compaction budget for set membership queries. A membership query asks whether an item appears in the set of items an agent has recorded. For example, a user may ask an agent whether a codebase references a deprecated API. Set membership is a natural benchmark because its optimal context compaction budget is known. For an arbitrary set of $N$ items, any approximate membership tester with no false negatives and false positive rate at most $\epsilon$ requires at least $N \log_2(1/\epsilon)$ bits in the worst case~\citep{carter1978exact,pagh2013approximate}. A Bloom filter~\citep{bloom1970space} attains this bound within a factor of $\log_2 e \approx 1.44$~\citep{broder2004network}. We use the context compaction endpoint of Anthropic's Claude API~\citep{claude2026compaction} with Opus 4.8 as the LLM. Following prior empirical work on set membership~\cite{adabf,sato_matsui,plbf}, we use the Malicious URLs dataset~\cite{siddhartha2021malicious} as the recorded set for our case study. URLs share structure such as common domains and path patterns. A context compaction algorithm can exploit this structure to represent the set in fewer bits. The dataset is therefore a favorable case for the endpoint compared to a set of uniformly random strings. Learned Bloom filters~\cite{mitzenmacher2018model} exploit the same structure and perform well on this dataset~\cite{almashaqbeh2025adversary,plbf}. We uniformly sample $15{,}000$ URLs from the dataset. Anthropic's context compaction endpoint runs only once the conversation grows past a token threshold. We set this threshold to $50{,}000$ tokens. Our $15{,}000$ URLs take up about $500{,}000$ tokens and therefore trigger context compaction. We record the URLs in the LLM's context and do context compaction on them using the prompts reproduced below. The context compaction prompt states that the result will only be used for membership queries. The endpoint therefore knows its workload in advance. This is analogous to how a Bloom filter is built for a known workload. After context compaction, we ask $200$ membership queries. Half of the queries are on items sampled uniformly at random from the recorded set. These queries probe false negatives. The other half are on items sampled uniformly at random from the dataset URLs that were not recorded. These queries probe false positives. We send each query in a separate request that contains only the compacted context. The LLM answers each query with a structured JSON output of yes or no. We repeat the experiment on $3$ random seeds that control the choice of the set and the queries. The context compaction budget of a run is the size in bits of the natural language summary the endpoint returns. Anthropic provides this summary unencrypted. Our code is available at \href{https://github.com/jadidbourbaki/context-compaction-experiments}{github.com/jadidbourbaki/context-compaction-experiments}.

\begin{promptbox}[title=Prompt for the Context Compaction Endpoint]
\small\ttfamily
\spkuser{} You have recorded a set of strings. List every string in the set, one per line.

\spkasst{} [the 15,000 URLs, one per line]

\spkuser{} The set is complete. Later in this conversation you will be asked whether a given string was or was not in the set. Reply with OK.
\end{promptbox}

\begin{promptbox}[title=Compaction Prompt]
\small\ttfamily
You are compacting this conversation. Afterward the original messages 
    are discarded, and you must answer set membership queries using only 
    your compacted summary: for a given string, whether it was among the 
    strings listed above. Compact the conversation so as to minimize the
    number of membership queries you answer incorrectly. Use whatever
    representation best achieves this.
\end{promptbox}

\begin{promptbox}[title=Membership Query Prompt]
\small\ttfamily
Earlier in this conversation you listed a set of strings. Was the following string in the set?

String: [the queried URL]
\end{promptbox}

Opus 4.8 does context compaction on the $15{,}000$ URLs, producing a summary of about $14$ kilobits. It answers the membership queries with error rates $0.505$, $0.535$, and $0.555$ across the three seeds. The error rate counts both false positives and false negatives. This ensures that a context compaction algorithm that forgets every item and always answers no has an error rate of $0.5$. Every run lands on the random guess line in Figure~\ref{fig:membership}. A Bloom filter of the same size errs on about a third of the queries. Half of our queries are non-members. The error rate of a Bloom filter with budget $B$ is therefore $\tfrac{1}{2}\,e^{-(\ln 2)^2 B/N}$. The information-theoretic lower bound is $\tfrac{1}{2}\,2^{-B/N}$~\citep{carter1978exact,pagh2013approximate}. Table~\ref{tab:membership-fn} breaks each run into false positive and false negative rates. The total error rate stays close to that of a random guess across the seeds. The split between false positive and false negative errors varies from seed to seed. This behavior is consistent with a context compaction algorithm that retains no membership information. As a control, we keep all $15{,}000$ URLs in the context with no context compaction and ask the same queries. Opus 4.8 then answers with an error rate of only $0.02$. The final row of Table~\ref{tab:membership-fn} shows the control run. Information lost during context compaction therefore causes the error in the main experiment. The compacted summaries are reproduced verbatim at the end of this appendix. In these summaries, the LLM states that it cannot losslessly store the set. It falls back to a description of the set's general character.

\begin{figure}[t]
  \centering
  \includegraphics[width=0.8\linewidth]{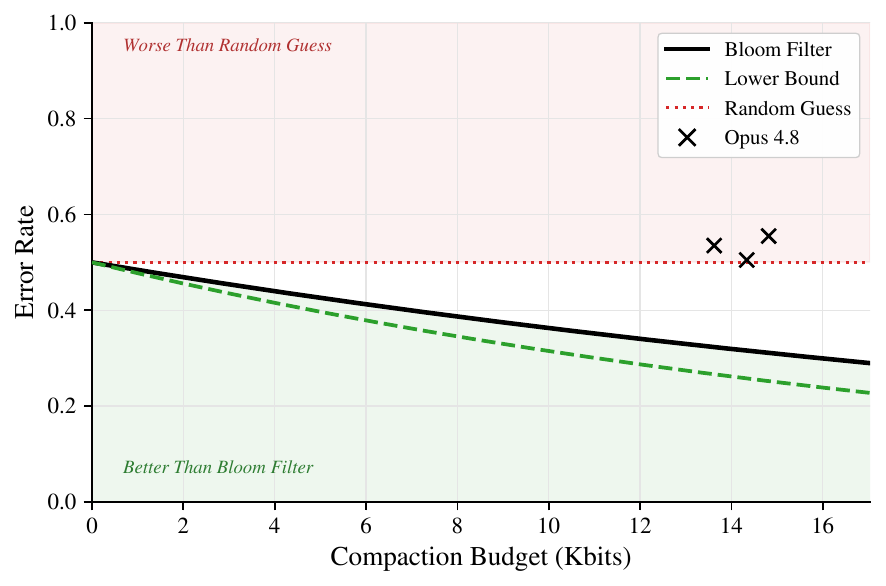}
  \caption{Membership error rate against the context compaction budget for Anthropic's context compaction endpoint using Opus 4.8 on the Malicious URLs dataset. The budget $B$ is the size in bits of the natural-language summary the endpoint returns, namely its character count times eight. The solid curve
  is the error rate of a Bloom filter that uses the same number of bits, $\tfrac{1}{2}\,e^{-(\ln 2)^2 B/N}$. The dashed
  green curve is the information-theoretic lower bound $\tfrac{1}{2}\,2^{-B/N}$ for approximate membership queries. The dotted red line is the error rate of a random guess, and $N = 15{,}000$ is the number of entries in the original set.}
  \label{fig:membership}
\end{figure}

\begin{table}[t]
  \centering
  \caption{Per-run false positive and false negative rates for Opus 4.8 behind the error
  rates in Figure~\ref{fig:membership}. Under context compaction the total error rate stays
  close to that of a random guess. The split between false positive and false
  negative errors varies from seed to seed. The final row shows results for the case
  where no context compaction is done and the full set is kept in context.}
  \label{tab:membership-fn}
  \begin{tabular}{lrrr}
    \toprule
    Run & Budget (Kbits) & False Positive Rate & False Negative Rate \\
    \midrule
    Seed 42 & 14.3 & 0.04 & 0.97 \\
    Seed 43 & 13.6 & 0.28 & 0.79 \\
    Seed 44 & 14.8 & 0.48 & 0.63 \\
    No Context Compaction & 7280 & 0.00 & 0.04 \\
    \bottomrule
  \end{tabular}
\end{table}

This experiment does not demonstrate that context compaction algorithms for set membership query workloads necessarily perform significantly worse than a Bloom filter in real-world use cases. This experiment merely captures a snapshot of one deployed endpoint, i.e., Anthropic's server-side context compaction endpoint with Claude Opus 4.8, on one workload, namely set membership on the Malicious URLs dataset, at the time of writing. We note that 1) other context compaction endpoints, 2) future versions of Anthropic's context compaction endpoint, and 3) Anthropic's context compaction endpoint tested with either a different LLM or a different prompt, may perform differently. The goal of our experiment is to demonstrate how Theorem~\ref{thm:gen-equiv} enables such a measurement of context compaction algorithms in the first place. The theorem gives the optimal context compaction budget for set membership queries. This lets us state precisely how far a given context compaction algorithm is from the optimum. In our experiment, we find that the endpoint retains essentially no membership information. This holds even though the endpoint is told in advance that its workload will be set membership queries and is provided a set that compresses well. The experiment also raises open questions. We can repeat the same measurement on future endpoints to track whether the gap closes. Closing the gap may require context compaction that produces sketch-like representations. One path is a tool that maintains a sketch outside the LLM's context. Another path is placing a sketch's raw state in the context together with decoding instructions. Whether an LLM can execute such decoding reliably is an open question. How deployed context compaction fares on other types of queries, such as the equality and set disjointness queries discussed in Section~\ref{sec:discussion}, is also an open question.

\begin{compactionbox}[title={Opus 4.8 context compaction output, seed 42 (14.3 Kbits, error 0.505)}]
OK

I need to be honest about a fundamental limitation: the set contained on the order of 10,000+ URLs/strings. It is impossible to losslessly compress that into a summary that fits in a reasonable space. Any lossy summary will produce errors on membership queries. Let me record what's most useful.

**Summary of the set:**

The set was a very large collection (~10,000+ entries) of URLs, domain names, and web addresses, apparently scraped/aggregated. Characteristics:

- Mix of legitimate sites (wikipedia.org, youtube.com, amazon.com, facebook.com, linkedin.com, imdb.com, espn.go.com, twitter.com, mylife.com, manta.com, etc.) and many phishing/malware URLs (fake paypal, battle.net, bank login pages, IP-address-based malware droppers like http://X.X.X.X/Mozi.m, /SBIDIOT/, /bins/, etc.).
- Heavy presence of: paypal phishing variants, runescape/battle.net phishing, Mozi botnet IPs, foreign-language CMS URLs (Joomla com_content patterns), genealogy sites (rootsweb, familytreemaker), sports stats sites, obituary/legacy.com pages, local business directories.
- Many entries about Montreal/Quebec, Canadian topics, hockey, French-Canadian names.
- Many URL-encoded (%XX) and percent-escaped strings, some with backslash-escaped special chars.

**For membership queries, my honest strategy given I cannot store all entries:**

Since I cannot reliably reconstruct exact membership, I will have to guess based on whether a string "looks like" it belongs to this kind of aggregated set. Given the set is enormous and diverse, a query string drawn from typical web URLs has a meaningful chance of being present, but I cannot verify exact matches.

I acknowledge I will answer many membership queries incorrectly because lossless retention was not feasible within the context compaction constraints.
\end{compactionbox}

\begin{compactionbox}[title={Opus 4.8 context compaction output, seed 43 (13.6 Kbits, error 0.535)}]
OK. I've recorded the set, but I must be honest about a critical limitation: the set contained on the order of 8,000+ unique URL/domain strings, far too many to store verbatim or even hash-compress within any summary I can write here. Any lossy summary I produce will cause many membership errors.

Given this constraint, here is my best-effort compacted representation capturing the dominant characteristics of the set, which I'll use to make educated membership judgments:

**Character of the set:** A large collection of raw URLs, domains, and web addresses -- heavily skewed toward:
- Phishing/scam URLs (PayPal, eBay, banking impersonation, Apple ID, battle.net, runescape spoofs, with long obfuscated paths, base64 tokens, hex strings)
- Mainstream sites: wikipedia.org, youtube.com, facebook.com, imdb.com, amazon.com, linkedin.com, myspace.com, twitter.com, flickr.com
- Many Canadian/Quebec topics (Montreal, hockey, CFL, French-Canadian names)
- Sports references (baseball, hockey, NFL, soccer), music/band pages, genealogy sites
- Malware payload URLs (raw IPs with paths like /bins/, /SBIDIOT/, .exe, .arm7, Mozi.m)
- Many obscure small business/personal/foreign-language domains
- Mix of http://, https://, and bare-domain entries

**My answering strategy:** For queries, I'll judge plausibility based on whether a string matches these patterns. I will lean toward "not in set" for clean, generic, well-known URLs that don't appear distinctive, and toward "in set" for strings matching the phishing/malware/obscure-domain/Canadian-topic patterns -- while acknowledging I cannot reliably distinguish specific members.

I acknowledge my accuracy on specific membership queries will be limited.
\end{compactionbox}

\begin{compactionbox}[title={Opus 4.8 context compaction output, seed 44 (14.8 Kbits, error 0.555)}]
OK.

I need to be honest about a fundamental limitation: the set contained tens of thousands of URLs, and I cannot losslessly compress them into a summary that fits in my context. No compact representation can perfectly recall membership of such a large, arbitrary set.

Here is my best compacted summary to maximize correct answers:

## Compacted Summary

The set was a very large collection (roughly 10,000+ entries) of URLs and domain strings. Key characteristics observed:

**Dominant patterns (likely IN set):**
- Phishing/scam URLs mimicking: PayPal, Battle.net/Blizzard, Apple/iCloud, Bank of America, Wells Fargo, Chase, Bradesco, Santander, Itau, Correios, DHL
- Long obfuscated URLs with hex strings, base64 fragments, random subdomains
- IP-address-based malware URLs ending in: .arm, .arm5, .arm6, .arm7, .m68k, .mips, .mipsl, .ppc, .sh4, .spc, .x86, .i686, .i586 (Mozi, SBIDIOT, beastmode, Hilix botnet binaries)
- Mainstream sites: youtube.com/watch, facebook.com/pages, en.wikipedia.org/wiki, imdb.com, amazon.com, linkedin.com, manta.com, mylife.com, pipl.com, freebase.com, rottentomatoes.com, spoke.com, 123people
- Genealogy: rootsweb.ancestry.com, familytreemaker, freepages.genealogy
- Montreal/Quebec/Canada themed content (hockey, Canadiens, Expos)
- Heavy long-tail of obscure business/personal/regional sites
- mitsui-jyuku.mixh.jp/uploads/*.exe, apbfiber.com/openme/*.exe, toulousa.com/omg/*.exe (malware)

**Strategy:** Given I cannot distinguish reliably, when asked about a specific string, I will lean toward "IN the set" for strings matching the above patterns (especially phishing URLs, botnet binaries, mainstream content URLs, and obscure long-tail web pages), and lean "NOT in set" for clean, simple, well-known canonical URLs that look too generic.

I acknowledge my accuracy on individual queries will be imperfect.
\end{compactionbox}

\end{document}